\documentclass[a4paper,conference]{IEEEtran}
\usepackage{graphicx,ifthen}
\usepackage{psfrag,amssymb,amsthm,subfigure}
\usepackage[cmex10]{amsmath}
\usepackage{cite,flushend,color,pst-plot,stfloats,pst-sigsys}
\usepackage{pstricks-add}

\newtheorem{thm}{Theorem}
\newtheorem{lem}{Lemma}
\newtheorem{cor}{Corollary}

\newtheorem{prop}{Proposition}

\theoremstyle{definition}
\newtheorem{definition}{Definition}

\def \arxiv {1}

\title{Signal Enhancement as Minimization of Relevant Information Loss}

\author{\IEEEauthorblockN{Bernhard C. Geiger\IEEEauthorrefmark{1}, Gernot Kubin\IEEEauthorrefmark{1}
\IEEEauthorblockA{\IEEEauthorrefmark{1}Signal Processing and Speech Communication Laboratory, Graz University of Technology, Austria}
$\{$geiger,gernot.kubin$\}$@tugraz.at}}

\begin{document}
\newcounter{myTempCnt}

\ifthenelse{\arxiv=1}{
\newcommand{\x}[1]{x[#1]}
\newcommand{\y}[1]{y[#1]}

\newcommand{\pdfy}{f_Y(y)}

\newcommand{\ent}[1]{H(#1)}
\newcommand{\diffent}[1]{h(#1)}
\newcommand{\derate}[1]{\bar{h}\left(\mathbf{#1}\right)}
\newcommand{\mutinf}[1]{I(#1)}
\newcommand{\ginf}[1]{I_G(#1)}
\newcommand{\kld}[2]{D(#1||#2)}
\newcommand{\kldrate}[2]{\bar{D}(\mathbf{#1}||\mathbf{#2})}
\newcommand{\binent}[1]{H_2(#1)}
\newcommand{\binentneg}[1]{H_2^{-1}\left(#1\right)}
\newcommand{\entrate}[1]{\bar{H}(\mathbf{#1})}
\newcommand{\mutrate}[1]{\mutinf{\mathbf{#1}}}
\newcommand{\redrate}[1]{\bar{R}(\mathbf{#1})}
\newcommand{\pinrate}[1]{\vec{I}(\mathbf{#1})}
\newcommand{\loss}[2][\empty]{\ifthenelse{\equal{#1}{\empty}}{L(#2)}{L_{#1}(#2)}}
\newcommand{\lossrate}[2][\empty]{\ifthenelse{\equal{#1}{\empty}}{\overline{L}(\mathbf{#2})}{L_{\mathbf{#1}}(\mathbf{#2})}}
\newcommand{\gain}[1]{G(#1)}
\newcommand{\atten}[1]{A(#1)}
\newcommand{\relLoss}[2][\empty]{\ifthenelse{\equal{#1}{\empty}}{l(#2)}{l_{#1}(#2)}}
\newcommand{\relLossrate}[1]{l(\mathbf{#1})}
\newcommand{\relTrans}[1]{t(#1)}
\newcommand{\partEnt}[2]{H^{#1}(#2)}

\newcommand{\dom}[1]{\mathcal{#1}}
\newcommand{\indset}[1]{\mathbb{I}\left({#1}\right)}

\newcommand{\unif}[2]{\mathcal{U}\left(#1,#2\right)}
\newcommand{\chis}[1]{\chi^2\left(#1\right)}
\newcommand{\chir}[1]{\chi\left(#1\right)}
\newcommand{\normdist}[2]{\mathcal{N}\left(#1,#2\right)}
\newcommand{\Prob}[1]{\mathrm{Pr}(#1)}
\newcommand{\Mar}[1]{\mathrm{Mar}(#1)}
\newcommand{\Qfunc}[1]{Q\left(#1\right)}

\newcommand{\expec}[1]{\mathrm{E}\left\{#1\right\}}
\newcommand{\expecwrt}[2]{\mathrm{E}_{#1}\left\{#2\right\}}
\newcommand{\var}[1]{\mathrm{Var}\left\{#1\right\}}
\renewcommand{\det}{\mathrm{det}}
\newcommand{\cov}[1]{\mathrm{Cov}\left\{#1\right\}}
\newcommand{\sgn}[1]{\mathrm{sgn}\left(#1\right)}
\newcommand{\sinc}[1]{\mathrm{sinc}\left(#1\right)}
\newcommand{\e}[1]{\mathrm{e}^{#1}}
\newcommand{\multint}{\iint{\cdots}\int}
\newcommand{\modd}[3]{((#1))_{#2}^{#3}}
\newcommand{\quant}[1]{Q\left(#1\right)}
\newcommand{\card}[1]{\mathrm{card}(#1)}
\newcommand{\diam}[1]{\mathrm{diam}(#1)}
\newcommand{\rec}[1]{r(#1)}
\newcommand{\recmap}[1]{r_{\mathrm{MAP}}(#1)}

\newcommand{\ivec}{\mathbf{i}}
\newcommand{\hvec}{\mathbf{h}}
\newcommand{\gvec}{\mathbf{g}}
\newcommand{\avec}{\mathbf{a}}
\newcommand{\kvec}{\mathbf{k}}
\newcommand{\fvec}{\mathbf{f}}
\newcommand{\vvec}{\mathbf{v}}
\newcommand{\xvec}{\mathbf{x}}
\newcommand{\Xvec}{\mathbf{X}}
\newcommand{\Xhvec}{\hat{\mathbf{X}}}
\newcommand{\xhvec}{\hat{\mathbf{x}}}
\newcommand{\xtvec}{\tilde{\mathbf{x}}}
\newcommand{\Yvec}{\mathbf{Y}}
\newcommand{\yvec}{\mathbf{y}}
\newcommand{\Zvec}{\mathbf{Z}}
\newcommand{\Svec}{\mathbf{S}}
\newcommand{\Nvec}{\mathbf{N}}
\newcommand{\Pvec}{\mathbf{P}}
\newcommand{\muvec}{\boldsymbol{\mu}}
\newcommand{\wvec}{\mathbf{w}}
\newcommand{\Wvec}{\mathbf{W}}
\newcommand{\Hmat}{\mathbf{H}}
\newcommand{\Amat}{\mathbf{A}}
\newcommand{\Fmat}{\mathbf{F}}

\newcommand{\zerovec}{\mathbf{0}}
\newcommand{\eye}{\mathbf{I}}
\newcommand{\evec}{\mathbf{i}}

\newcommand{\zeroone}{\left[\begin{array}{c}\zerovec^T\\ \eye\end{array} \right]}
\newcommand{\zerooneT}{\left[\begin{array}{cc}\zerovec & \eye\end{array} \right]}
\newcommand{\zerooneM}{\left[\begin{array}{cc}\zerovec &\zerovec^T\\\zerovec& \eye\end{array} \right]}

\newcommand{\Cxx}{\mathbf{C}_{XX}}
\newcommand{\Cx}{\mathbf{C}_{\Xvec}}
\newcommand{\Chx}{\hat{\mathbf{C}}_{\Xvec}}
\newcommand{\Cy}{\mathbf{C}_{\Yvec}}
\newcommand{\Cz}{\mathbf{C}_{\Zvec}}
\newcommand{\Cn}{\mathbf{C}_{\mathbf{N}}}
\newcommand{\Cnt}{\underline{\mathbf{C}}_{\tilde{\mathbf{N}}}}
\newcommand{\Cntm}{\underline{\mathbf{C}}_{\tilde{\mathbf{N}}}}
\newcommand{\Cxh}{\mathbf{C}_{\hat{X}\hat{X}}}
\newcommand{\rxx}{\mathbf{r}_{XX}}
\newcommand{\Cxy}{\mathbf{C}_{XY}}
\newcommand{\Cyy}{\mathbf{C}_{YY}}
\newcommand{\Cnn}{\mathbf{C}_{NN}}
\newcommand{\Cyx}{\mathbf{C}_{YX}}
\newcommand{\Cygx}{\mathbf{C}_{Y|X}}
\newcommand{\Wmat}{\underline{\mathbf{W}}}

\newcommand{\Jac}[2]{\mathcal{J}_{#1}(#2)}

\newcommand{\NN}{{N{\times}N}}
\newcommand{\perr}{P_e}
\newcommand{\perh}{\hat{\perr}}
\newcommand{\pert}{\tilde{\perr}}

\newcommand{\vecind}[1]{#1_0^n}
\newcommand{\roots}[2]{{#1}_{#2}^{(i_{#2})}}
\newcommand{\rootx}[1]{x_{#1}^{(i)}}
\newcommand{\rootn}[2]{x_{#1}^{#2,(i)}}

\newcommand{\markkern}[1]{f_M(#1)}
\newcommand{\pole}{a_1}
\newcommand{\preim}[1]{g^{-1}[#1]}
\newcommand{\preimV}[1]{\mathbf{g}^{-1}[#1]}
\newcommand{\Xmax}{\bar{X}}
\newcommand{\Xmin}{\underbar{X}}
\newcommand{\xmax}{x_{\max}}
\newcommand{\xmin}{x_{\min}}
\newcommand{\limn}{\lim_{n\to\infty}}
\newcommand{\limk}{\lim_{k\to\infty}}
\newcommand{\limX}{\lim_{\hat{\Xvec}\to\Xvec}}
\newcommand{\limx}{\lim_{\hat{X}\to X}}
\newcommand{\limXo}{\lim_{\hat{X}_1\to X_1}}
\newcommand{\sumin}{\sum_{i=1}^n}
\newcommand{\finv}{f_\mathrm{inv}}
\newcommand{\ejtheta}{\e{\jmath\theta}}
\newcommand{\khat}{\bar{k}}
\newcommand{\modeq}[1]{g(#1)}
\newcommand{\partit}[1]{\mathcal{P}_{#1}}
\newcommand{\psd}[1]{S_{#1}(\e{\jmath \theta})}
\newcommand{\borel}[1]{\mathfrak{B}(#1)}
\newcommand{\infodim}[1]{d(#1)}

\newcommand{\delay}[2]{\psblock(#1){#2}{\footnotesize$z^{-1}$}}
\newcommand{\Quant}[2]{\psblock(#1){#2}{\footnotesize$\quant{\cdot}$}}
\newcommand{\moddev}[2]{\psblock(#1){#2}{\footnotesize$\modeq{\cdot}$}}}{}

\renewcommand{\Cn}{\mathbf{\underline{C}_N}}
\renewcommand{\Cx}{\mathbf{\underline{C}_X}}
\renewcommand{\Cz}{\mathbf{\underline{C}_S}}
\renewcommand{\Cy}{\mathbf{\underline{C}_Y}}

\maketitle

\begin{abstract}
We introduce the notion of relevant information loss for the purpose of casting the signal enhancement problem in information-theoretic terms. We show that many algorithms from machine learning can be reformulated using relevant information loss, which allows their application to the aforementioned problem. As a particular example we analyze principle component analysis for dimensionality reduction, discuss its optimality, and show that the relevant information loss can indeed vanish if the relevant information is concentrated on a lower-dimensional subspace of the input space.
\end{abstract}

\ifthenelse{\arxiv=1}{
\begin{figure*}[t]
 \centering  
 \subfigure[Computing the (absolute) information loss $\loss{X\to Y}$]{\label{fig:sysmod_abs}
 \begin{pspicture}[showgrid=false](1,1)(8.5,3.5)
      \psset{style=RoundCorners}
	\pssignal(4.5,2){xorig}{\textcolor{red}{$X$}}
	\psfblock[linecolor=red,framesize=1.5 1](6.5,2){c}{\textcolor{red}{$g(\cdot)$}}
	\pssignal(8.5,2){y}{\textcolor{red}{$Y$}}
	\nclist[style=Arrow,linecolor=red]{ncline}[naput]{xorig,c,y}
 	\pssignal(1,2){x}{$\hat{X}$}
	\pssignal(3,1){n}{$\partit{}$}
	\psfblock[framesize=1 0.75](3,2){oplus}{$Q(\cdot)$}
	\ncline[style=Arrow]{n}{oplus}
	\nclist[style=Arrow]{ncline}[naput]{xorig,oplus,x}

	\psline[style=Dash](1.1,2.75)(4.5,2.75)
	\psline[style=Dash](1.1,2.25)(1.1,2.75)
	\psline[style=Dash](4.5,2.25)(4.5,2.75)
	\psline[style=Dash](0.9,3.25)(8.5,3.25)
	\psline[style=Dash](0.9,2.25)(0.9,3.25)
	\psline[style=Dash](8.5,2.25)(8.5,3.25)
 	\rput*(2.8,2.75){\scriptsize{$\mutinf{\hat{X};X}$}}
	\rput*(4.7,3.25){\scriptsize{$\mutinf{\hat{X};Y}$}}
\end{pspicture}}
\quad
 \subfigure[Computing the relevant information loss $L_S({X\to Y})$]{\label{fig:sysmod_rel}
 \begin{pspicture}[showgrid=false](0,1)(7,5.5)
    \psset{style=RoundCorners}
    \pssignal(0.5,2){xorig}{\textcolor{red}{$X$}}
    \psfblock[linecolor=red,framesize=1.5 1](3.5,2){c}{\textcolor{red}{$g(\cdot)$}}
    \pssignal(6.5,2){y}{\textcolor{red}{$Y$}}
	\nclist[style=Arrow,linecolor=red]{ncline}[naput]{xorig,c,y}
	
	\pssignal(0.5,4.5){s}{{$S$}}
	\pssignal(2,5.5){n}{$\partit{}$}
	\psfblock[framesize=1 0.75](2,4.5){oplus}{$Q(\cdot)$}
	\ncline[style=Arrow]{n}{oplus}
	\pssignal(3.5,4.5){shat}{{$\hat{S}$}}
	\nclist[style=Arrow]{ncline}[naput]{s,oplus,shat}
\ncline[style=Arrow,style=Dash]{s}{xorig}
	\psline[style=Dash](0.5,2.25)(3.5,4.25)
	\psline[style=Dash](6.5,2.25)(3.5,4.25)
	\rput*(2,3.25){\scriptsize{$\mutinf{\hat{S};X}$}}
	\rput*(5,3.25){\scriptsize{$\mutinf{\hat{S};Y}$}}
\end{pspicture}}
\caption{Model for computing the information loss of a memoryless input-output system $g$. $Q$ is a quantizer with partition $\partit{}$.}
\label{fig:sysmod}
\end{figure*}
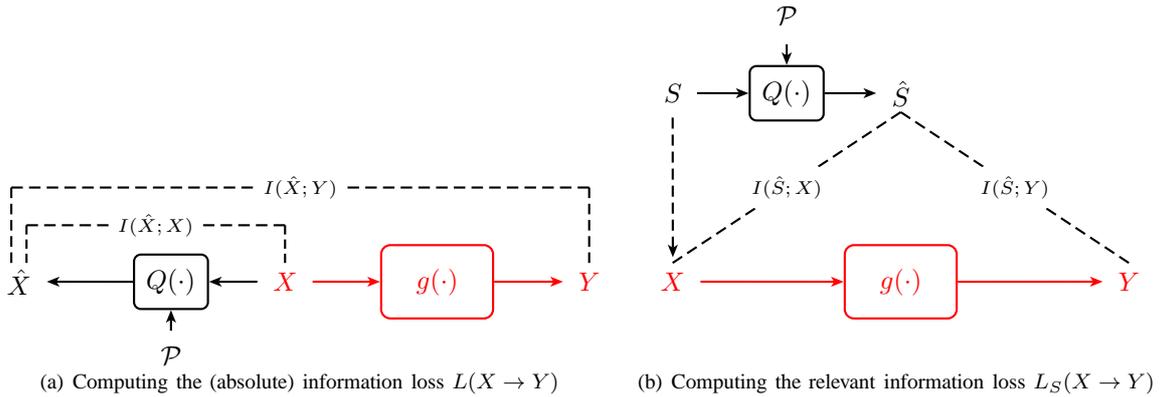}{}

\section{Introduction}\label{sec:intro}
It is a widely known fact from information theory, that processing a random variable, or a signal, cannot increase the amount of information it represents. In fact, as the data processing inequality states, information can only be lost by passing a signal through a deterministic system. This information loss can -- very loosely speaking -- be interpreted as the difference between the information at the input and the information at the output of the system. In case of continuous-valued random variables (which contain an infinite amount of information) the authors developed a theory quantifying the information loss both in absolute and relative terms~\cite{Geiger_ILStatic_IZS,Geiger_RILPCA_arXiv}.

We now look at the data processing inequality (DPI) from a different point of view by asking the following question: How can we justify \emph{signal processing}, knowing that it can only \emph{reduce} information? Whenever technical systems prepare a physical, information-carrying signal for human perception, processing occurs and, presumably, a significant amount of information is lost. The only justification for this approach is that the information which is lost is actually not the information we are interested in, but rather some nuisance. We hope to preserve all the information \emph{relevant} to us while minimizing the nuisance. The problem of signal enhancement, as we will argue later, is an optimization problem which can be cast in exactly these terms.\ifthenelse{\arxiv=1}{ Aside from that, each block in the signal processing chain has the goal of representing the relevant information such that as little as possible is lost in the subsequent blocks; as in the sense of~\cite{Johnson_ITNeural}, preprocessing can improve performance, while postprocessing cannot -- we cannot recover lost information, but we can prevent loosing it.}{} It is our goal to make these statements precise by supporting them with a mathematical framework of \emph{relevant information loss}.

Relevant information and its counterpart, relevant information loss, are concepts not altogether new. Indeed, the notion of relevance has a long history in machine learning and neural networks: Plumbley, essentially using the same definition as we do, explicitly used relevant information loss in analyzing the properties of principle component analysis~\cite{Plumbley_TN,Deco_ITNN}. The information bottleneck (IB) method and its variants directly maximize relevant information while minimizing the complexity of its signal representation~\cite{Tishby_InformationBottleneck}. Relevance w.r.t. a specific goal was the motivation for defining an information processing measure for neural networks in~\cite{Sanchez_ProcessingMeasure,Sanchez_PhD}. Finally, the principle of relevant information, being structurally similar to the information bottleneck formulation, builds on minimizing the relative entropy between the relevant random variable and its representation~\cite{Principe_ITLEarning}.

Conversely, in signal processing the notion of relevant information has not gained foothold yet. In this paper we thus formulate a definition of relevant information loss (Section~\ref{sec:defLoss}) following the ideas of~\cite{Plumbley_TN} and discuss its properties from the view-point of system theory (Section~\ref{sec:basicProp}). On the basis of these properties we then discuss the problem of signal enhancement in Section~\ref{sec:signal_enhancement} and investigate the connections to the machine learning literature. As a special example of signal enhancement we analyze principle component analysis (Section~\ref{sec:PCA}), generalizing the results of~\cite{Plumbley_TN}. Finally, we give a brief example of the application of our results to a simple digital communication system (Section~\ref{sec:examples}).

\ifthenelse{\arxiv=1}{}{An extended version of this paper, including all proofs omitted here for the sake of brevity, is available in~\cite{Geiger_Relevant_arXiv}.}

\section{A Definition of Relevant Information Loss}\label{sec:defLoss}
We start with recalling the definition given in~\cite{Geiger_ILStatic_IZS}, where the information loss induced by transforming a (possibly multidimensional) random variable (RV) $X$ to another RV $Y$ by a static function $g{:}\ \dom{X}\to\dom{Y}$, $\dom{X},\dom{Y}\subseteq\mathbb{R}^N$ was given as
\begin{equation}
 \loss{X\to Y} = \sup_{\partit{}} \left(\mutinf{\hat{X};X}-\mutinf{\hat{X};Y}\right) = \ent{X|Y}\label{eq:loss}
\end{equation}
where the supremum is over all partitions $\partit{}$ of the sample space $\dom{X}$ of $X$, and where $\hat{X}$ is obtained by quantizing $X$ accordingly\ifthenelse{\arxiv=1}{ (see Fig.~\ref{fig:sysmod_abs})}{}. This definition of absolute information loss is accompanied by an expression for the relative information loss in~\cite{Geiger_RILPCA_arXiv}. 

As the examples in~\cite{Geiger_RILPCA_arXiv} show, both the absolute and the relative definition of information loss have their shortcomings, especially when it comes to systems $g$ used for signal enhancement: Since the expressions only consider the RV $X$ at the input of the system, they do not take into account that not all of the information contained in $X$ is \emph{relevant}, often leading to counter-intuitive results. The main contribution of this work lies thus in analyzing the implications of the following

\begin{definition}[Relevant Information Loss]\label{def:relevantLoss}
 Let $X$ be an RV on the sample space $\dom{X}$, and let $Y$ be obtained by transforming $X$ with a static function $g$. Let $S$ be another RV on the sample space $\dom{S}$ representing \emph{relevant information}. The information loss \emph{relevant w.r.t. $S$} is defined as
\begin{equation}
 \loss[S]{X\to Y} = \sup_{\partit{}} \left(\mutinf{\hat{S};X}-\mutinf{\hat{S};Y}\right)=\mutinf{X;S|Y}.\label{eq:relevantloss}
\end{equation}
where the supremum is over all partitions of the sample space $\dom{S}$, and where $\hat{S}$ is obtained by quantizing $S$ accordingly.
\end{definition}

The information loss relevant w.r.t. an RV $S$ is thus the difference of mutual informations between $S$ and the input and output of the system\ifthenelse{\arxiv=1}{ (see Fig.~\ref{fig:sysmod_rel})}{}. Equivalently, the relevant information loss is exactly the information that $X$ contains about $S$ \emph{which is not contained in $Y$}. Due to the data processing inequality~(e.g.,~\cite{Cover_Information2}) it is a non-negative quantity, i.e., a deterministic system cannot increase the amount of available relevant information. 

As we have already pointed out in the introduction, this definition of relevant information loss is not altogether new: Plumbley already introduced this quantity (named $\Delta I_S(X;Y)$, and omitting the supremum assuming finite mutual information between $S$ and $X$) in the context of unsupervised learning in neural networks~\cite{Plumbley_TN}. The rationale for this new measure was to circumvent some shortcomings of Linsker's principle of information maximization~\cite{Linsker_Infomax}: While infomax works well for Gaussian RVs and linear systems, applying the same algorithms to non-Gaussian data just maximizes an upper bound on the information. Plumbley's information loss, conversely, also yields closed form solutions for Gaussian data, but in addition to that one can derive upper bounds on the relevant information loss whenever the data is non-Gaussian. Thus, minimizing an upper bound on the information loss can be assumed to be more promising than maximizing an upper bound on the transferred information~\cite{Plumbley_TN}.

\section{Properties of Relevant Information Loss}\label{sec:basicProp}
We now analyze the elementary properties of relevant information loss:

\begin{prop}[Elementary Properties]\label{prop:elementary}
 The relevant information loss from Definition~\ref{def:relevantLoss} satisfies the following properties:
\begin{enumerate}
 \item $\loss[S]{X\to Y}\leq\mutinf{S;X}\leq\ent{S}$
 \item $\loss[S]{X\to Y}\geq\loss[Y]{X\to Y}=0$
 \item $\loss[S]{X\to Y}\leq\loss[X]{X\to Y}=\loss{X\to Y} $, with equality if $X$ is a function of $S$.
 \item $\loss[S]{X\to Y}=\ent{S|Y}$ if $S$ is a function of $X$.
\end{enumerate}
\end{prop}

\begin{proof}
 The first property results immediately from the definition, while the second property is due to the fact that $Y$ is a function of $X$. The third property results from making $X$ the relevant RV, thus making~\eqref{eq:relevantloss} equal to~\eqref{eq:loss}. Note that, with~\eqref{eq:loss},
\begin{IEEEeqnarray}{RCL}
 \loss[S]{X\to Y}=\mutinf{S;X|Y}  &=& \ent{X|Y}-\ent{X|Y,S}\\
&\leq& \loss{X\to Y}
\end{IEEEeqnarray}
with equality if $X$ is a function of $S$. The last property, for $S=f(X)$, follows by expanding $\mutinf{X;S|Y}$ as $\ent{S|Y}-\ent{S|Y,X} =\ent{S|Y}$.
\end{proof}

The third property is of particular interest. Essentially, it states that the relevant information loss cannot exceed the total information loss, so upper bounds (e.g., those presented in~\cite{Geiger_ILStatic_IZS}) for the latter can be used for the former as well.\ifthenelse{\arxiv=1}{ In addition to that, it suggests an alternative way to define $\loss{X\to Y}$: Namely as
\begin{equation}
 \loss{X\to Y} = \sup_{S} \left(\mutinf{S;X}-\mutinf{S;Y}\right)
\end{equation}
where the supremum is over all RVs on the sample space $\dom{X}$.}{}

Since by Definition~\ref{def:relevantLoss} relevant information loss is represented by a conditional mutual information, it inherits all of its properties. In particular, as we show next, a data processing inequality (DPI) holds:

\begin{prop}[Data Processing Inequality]\label{prop:dpi}
 Let $V-W-X-Y$ be a Markov chain. Then,
\begin{equation}
 \loss[W]{X\to Y}\geq \loss[V]{X\to Y}.
\end{equation}
\end{prop}

\begin{proof}
\ifthenelse{\arxiv=1}{See Appendix.}{See~\cite{Geiger_Relevant_arXiv}.}
\end{proof}

The particular usefulness of the DPI (especially in the context of our work) relies on the fact that both $S-X-g(X)$ and $f(S)-S-X$ are Markov chains. Comparing this to Proposition~\ref{prop:dpi} one is tempted to believe that the direction of the inequality depends on the fact whether $f(S)-S-X-Y$ or $S-f(S)-X-Y$ is a Markov chain. The following Corollary resolves this complication.

\begin{cor}\label{cor:dpi}
 Let $f$ be a measurable function defined on the sample space of $S$. Then,
\begin{equation}
 \loss[S]{X\to Y}\geq \loss[f(S)]{X\to Y}
\end{equation}
with equality if $S-f(S)-X-Y$ is a Markov chain.
\end{cor}

\begin{proof}
\ifthenelse{\arxiv=1}{See Appendix.}{Using Definition~\ref{def:relevantLoss}, the Corollary follows immediately from~\cite[Thm.~3.7.1]{Pinsker_InfoEngl}.}
\end{proof}
Before proceeding, we note that the third property of Proposition~\ref{prop:elementary} is an immediate consequence of this corollary.

We next show that, inherited from the properties of mutual information, also the relevant information loss obeys a chain rule:
\begin{prop}[Chain Rule of Information Loss]\label{prop:chain_rule}
 The information loss $\loss[S_1^n]{X\to Y}$ w.r.t. a collection $S_1^n=\{S_1,\dots,S_n\}$ of relevant RVs induced by a function $g$ satisfies
\begin{equation}
 \loss[S_1^n]{X\to Y}= \sum_{i=1}^n \loss[S_i|S_1^{i-1}]{X\to Y}.
\end{equation}
\end{prop}

The proof follows immediately from the chain rule of (conditional) information and is thus omitted. However, this chain rule justifies our intuitive understanding of the nature of information loss. We thus emphasize

\begin{cor}\label{cor:splitting}
  The information loss $\loss{X\to Y}$ induced by a function $g$ can be split into \emph{relevant} (w.r.t. $S$) and \emph{irrelevant} information loss:
\begin{equation}
 \loss{X\to Y} = \loss[S]{X\to Y}+\loss[X|S]{X\to Y}
\end{equation}
\end{cor}

\begin{proof}
 The proof follows from Proposition~\ref{prop:chain_rule} and from the fact that since $S$ and $Y$ are conditionally independent given $X$ we have $\loss[XS]{X\to Y} = \loss{X\to Y}$.
\end{proof}

\ifthenelse{\arxiv=1}{Note that even in a simple scenario with additive noise, i.e., $X=S+N$, it is not straightforward to just identify the noise $N$ with the irrelevant information. We will show this in one of the examples in Section~\ref{sec:examples}.
}{}

\begin{figure}[t]
\centering
  \begin{pspicture}[showgrid=false](0.5,1.5)(8,4)
	\psset{style=RoundCorners}
 	\pssignal(0.5,2){x}{$X$}\pssignal(0.5,3.5){s}{$S$}\psline[style=Dash,style=Arrow](0.5,3.25)(0.5,2.25)
	\psfblock[framesize=1.5 1](3,2){d}{\parbox{1.5\psunit}{\centering $\mathsf{System\ 1}$\\$g$}}
	\psfblock[framesize=1.5 1](6,2){c}{\parbox{1.5\psunit}{\centering $\mathsf{System\ 2}$\\$h$}}
	\pssignal(8,2){y}{$Z$}
  \nclist[style=Arrow]{ncline}[naput]{x,d,c $Y$,y}
	\psline[style=Dash](2,3)(4,3)
	\psline[style=Dash](2,2.5)(2,3)
	\psline[style=Dash](4,2.5)(4,3)
	\psline[style=Dash](5,3)(7,3)
	\psline[style=Dash](5,2.5)(5,3)
	\psline[style=Dash](7,2.5)(7,3)
	\psline[style=Dash](1.5,3.5)(7.5,3.5)
	\psline[style=Dash](1.5,2.5)(1.5,3.5)
	\psline[style=Dash](7.5,2.5)(7.5,3.5)
 	\rput*(3,3){\scriptsize\textcolor{black}{$\loss[S]{X\to Y}$}}
	\rput*(6,3){\scriptsize\textcolor{black}{$\loss[S]{Y\to Z}$}}
	\rput*(4.5,3.5){\scriptsize\textcolor{black}{$\loss[S]{X\to Z}$}}
\end{pspicture}
\caption{Cascade of two systems: The relevant information loss of the cascade is the sum of relevant information losses of the constituent systems.}
\label{fig:cascade}
\end{figure}
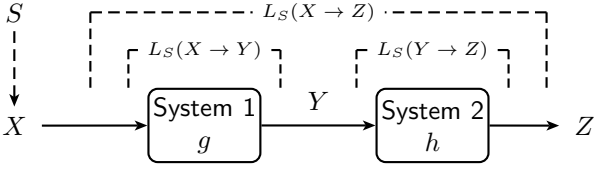

A final property we want to present concerns the cascade of systems (see Fig.~\ref{fig:cascade}). Assume, for the moment, that the output $Y$ of the first system $g$ is used as the input to a second system $h$, which responds with the RV $Z$. In~\cite{Geiger_ISIT2011arXiv} it was shown that the information loss of this cascade is additive, i.e., 
\begin{equation}
 \loss{X\to Z} = \loss{X\to Y} +\loss{Y\to Z}.
\end{equation}
The same holds for the relevant information loss, as Plumbley showed in~\cite{Plumbley_TN}. We thus state
\begin{prop}[Information Loss in a Cascade is Additive,~\cite{Plumbley_TN}]\label{prop:cascade}
 Let $Y$ be the RV obtained by transforming $X$ with a static function $g$, and let $Z$ be obtained by transforming $Y$ with $h$. Then, the information loss relevant w.r.t. $S$ is given as
\begin{equation}
 \loss[S]{X\to Z} = \loss[S]{X\to Y} +\loss[S]{Y\to Z}.
\end{equation}
\end{prop}

\begin{proof}
\ifthenelse{\arxiv=1}{See Appendix.}{See~\cite{Geiger_Relevant_arXiv,Plumbley_TN}.}
\end{proof}

\section{Signal Enhancement, Relevant Information Loss, and IB with Side Information}\label{sec:signal_enhancement}
\begin{figure}[t]
 \centering  
\begin{pspicture}[showgrid=false](2,1.5)(8,3.5)
  \psset{style=RoundCorners}
 	\pssignal(2,2){x}{$X$}\pssignal(2,3.5){s}{$S$}
	\psfblock[framesize=2 1.25](5,2){c}{\parbox[c]{1.5\psunit}{\centering $\mathsf{Signal}$\\ $\mathsf{Enhancer}$}}
	\pssignal(8,2){y}{$Y$}
	\ncline[style=Arrow]{x}{c}\ncline[style=Arrow]{c}{y}	
	\psline[style=Dash,style=Arrow](2,3.25)(2,2.25)
\end{pspicture}
\caption{The problem of signal enhancement: Given a (relevant) information signal $S$ and its observation $X$, how shall we design a system such that its output, $Y$, represents $S$ as good as possible?}
\label{fig:sigenhance}
\end{figure}
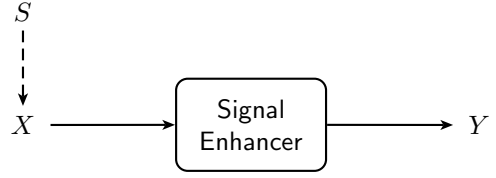

We now turn our attention to the signal enhancement problem. Often, the information one wants to retrieve is not directly available, but only through some corrupted observation: Information-carrying signals are superimposed by noise, distorted through nonlinear systems, and affected by time-varying, dynamic effects. It is essentially the goal of the signal processing engineer to mitigate all these adverse effects; to improve the quality of the observation such that as much information as possible can be retrieved from it with little effort.

Looking at this task from the perspective of information theory, we know that, as the data processing inequality dictates, signal enhancement does \emph{not} mean that one increases the amount of information in the observation. At best, one can build a system which \emph{preserves} as much information as possible (see Fig.~\ref{fig:sigenhance}). Conversely, noise, distortion, and other \emph{irrelevant} components of the observation should be removed such that information retrieval can be done easily. 

This is where our notion of relevant information loss comes in: If we let $S$ be the information carrying signal and $X$ its corrupted observation, our goal is to find a function $g$ such that the relevant information loss is minimized while simultaneously maximizing the irrelevant information loss. One may cast the resulting optimization problem as follows:
\begin{IEEEeqnarray}{RCL}
 \max_g &\quad& \loss[X|S]{X\to Y}\notag\\
 \text{s.t.} &&\loss[S]{X\to Y}\leq C\label{eq:signalEnhancement}
\end{IEEEeqnarray}
where $C$ is some constant\footnote{In some cases it may be beneficial to minimize the relevant information loss subject to a minimum reduction of irrelevant information, or even use a variational formulation.}. In other words, signal enhancement is the \emph{maximization of irrelevant information loss}.

Although in signal processing these information-theoretic considerations are only slowly gaining momentum, they are quite widespread in machine learning and neural networks, e.g.,~\cite{Deco_ITNN,Principe_ITLEarning}. In particular, a variational formulation of maximizing relevant information is the basis of the information bottleneck method (IB)~\cite{Tishby_InformationBottleneck}:
\begin{equation}
\min_{p(y|x)}\ \mutinf{Y;X}-\beta\mutinf{S;Y}\label{eq:IB}
\end{equation}
where the minimization is performed over all relations between the (discrete) RVs $Y$ and $X$ and where $\beta$ is a design parameter\ifthenelse{\arxiv=2}{\footnote{By the fact that $S-X-Y$ is a Markov chain and by the data processing inequality, $\mutinf{X;Y}\geq\mutinf{S;Y}$. As a consequence, equation~\eqref{eq:IB} is lower bounded by $(1-\beta)\mutinf{X;Y}$, which is a non-negative quantity for $\beta<1$. For this parameter choice the minimum is thus achieved for any constant $Y$ ($\mutinf{X;Y}=0$)~\cite{IBSI?}.}}{} trading compression for preservation of relevant information, $\mutinf{S;Y}$. While in principle the relation $p(y|x)$ can be stochastic, in many cases the algorithm is used for (hard) clustering using a deterministic function.\ifthenelse{\arxiv=1}{ We now try to express~\eqref{eq:IB} in terms of relevant and irrelevant information loss: Using $\loss[S]{X\to Y}=\mutinf{S;X}-\mutinf{S;Y}$ we obtain
\begin{IEEEeqnarray}{RCL}
 \mutinf{S;Y} &=& \mutinf{S;X}-\loss[S]{X\to Y}
\end{IEEEeqnarray}
where the first term is obviously independent of $p(y|x)$. Restricting ourselves to the clustering problem, i.e., to deterministic functions $Y=g(X)$, and letting $g^\circ$ be the optimal solution, we get
\begin{IEEEeqnarray}{RCL}
 g^\circ &=& \arg \min_g\ \mutinf{Y;X}-\beta\mutinf{S;Y}\\
&=& \arg \min_g\ -\ent{X}+\mutinf{Y;X}-\beta\mutinf{S;X}\notag\\&&{}+\beta\loss[S]{X\to Y}\\
&=& \arg \min_g\ -\loss{X\to Y}+\beta\loss[S]{X\to Y}\\
&=& \arg \min_g\ -\loss[S]{X\to Y}-\loss[X|S]{X\to Y}\notag\\&&{}+\beta\loss[S]{X\to Y}\\
&=& \arg \min_g\ (\beta-1)\loss[S]{X\to Y}-\loss[X|S]{X\to Y}\notag\\
\end{IEEEeqnarray}
and thus, the optimization problem can be cast as
\begin{equation}
 \min_g\ (\beta-1)\loss[S]{X\to Y}-\loss[X|S]{X\to Y}.
\end{equation}
}{
With Definition~\ref{def:relevantLoss} we can now write
\begin{multline}
 \mutinf{Y;X}-\beta\mutinf{S;Y} \\= \ent{X}-\underbrace{\ent{X|Y}}_{\loss{X\to Y}}
-\beta \mutinf{S;X}+\beta\loss[S]{X\to Y}.
\end{multline}
Neglecting all terms independent of $g$ and applying Corollary~\ref{cor:splitting} allows us to reformulate~\eqref{eq:IB} in terms of relevant and irrelevant information loss:
\begin{equation}
 \min_g\ (\beta-1)\loss[S]{X\to Y}-\loss[X|S]{X\to Y}
\end{equation}
}

Note that for large $\beta$ stronger emphasis is placed on minimizing relevant information loss, and a trivial solution to this problem is obtained by any bijective $g$. 

In the case of discrete RVs this problem can be circumvented by another variant of the IB method: The agglomerative IB method proposed in~\cite{Slonim_AgglomIB} starts with $Y=X$ and iteratively merges elements of the state space $\dom{Y}$ of $Y$ such that the relevant information loss is minimized in each step. By stopping this algorithm as soon as during some merging step $\loss[S]{X\to Y}>C$, at least a local optimum of our original signal enhancement problem~\eqref{eq:signalEnhancement} can be found.

A step further has been made by the authors of~\cite{Chechik_IBSI}, who extended the IB method to incorporate knowledge about irrelevant signal components. They accompany the relevant information $S$ by an irrelevance variable $\bar{S}$, assuming conditional independence between $S$ and $\bar{S}$ given $X$, and minimize the following functional:
\begin{equation}
 \mutinf{Y;X}-\beta\left[\mutinf{Y;S}-\gamma\mutinf{Y;\bar{S}}\right]
\end{equation}
At the same time, the entropy of the system output $Y$ and the irrelevant information $\mutinf{Y;\bar{S}}$ are minimized, and the relevant information $\mutinf{Y;S}$ is maximized. As in IB, $\beta$ and $\gamma$ are the weights for these three conflicting goals. We can again use our notions of relevant and, employing Corollary~\ref{cor:splitting}, irrelevant information loss to rewrite the optimization problem. If we restrict ourselves to clustering, one obtains
\begin{equation}
 \min_g\ (\beta-1)\loss[S]{X\to Y}-(\beta\gamma+1)\loss[X|S]{X\to Y}
\end{equation}
where again all constant terms have been dropped. Note that here $X|S$ takes the place of $\bar{S}$, automatically fulfilling the requirement of conditional independence in~\cite{Chechik_IBSI}.

Unlike for IB, by employing side information it does make sense to let $\beta\to\infty$: Compression emerges naturally from maximizing the irrelevant information loss (due to $\gamma>0$). This is exactly the approach that has been taken up by the authors of~\cite{Sanchez_ProcessingMeasure, Sanchez_PhD}, who introduced the following information processing measure\ifthenelse{\arxiv=2}{\footnote{A few things have to be considered there: First of all, they argue that the complexity reduction (irrelevant information loss) can be negative for stochastic systems, which has to be re-evaluated for our terminology. Secondly, they argue that not only the conditional entropy can act as an uncertainty measure, but, e.g., also Bayes error. Thirdly, by concentrating on entropy, they do not take into account differential entropy and thus quantize (cumbersomely) continuous RVs to make them accessible in their framework. Mutual information has not been taken into account by the authors. Finally, the authors argue that their cost function leads to PCA, Fisher discrimination, and C4.5. In their work, relevance is always defined w.r.t. some goal a (neural) processing system shall achieve.}}{} for discrete RVs $X$, $Y$, and $S$:
\begin{multline}
 \Delta P (X\to Y|S) = \ent{X|S}-\ent{Y|S}\\-\alpha\left(\ent{S|Y}-\ent{S|X}\right)\label{eq:sanchez}
\end{multline}
The authors argue that the first difference in~\eqref{eq:sanchez} corresponds to complexity reduction (i.e., reduction of irrelevant information), while the second term accounts for loss of relevant information. Indeed, with
\begin{IEEEeqnarray}{RCL}
 \loss[X|S]{X\to Y} &=& \mutinf{X;X|Y,S} = \ent{X|Y,S}\\
&=& \ent{X|S}-\ent{Y|S}
\end{IEEEeqnarray}
we can rewrite $\Delta P (X\to Y|S)$ with our notions of relevant information loss and state a variational utility function (to be maximized) for our signal enhancement problem~\eqref{eq:signalEnhancement}:
\begin{equation}
 \Delta P (X\to Y|S) = \loss[X|S]{X\to Y} -\alpha\loss[S]{X\to Y}
\end{equation}
Here, $\alpha>0$ is a design parameter trading between the loss of relevant and irrelevant information.

Aside from the IB method and its variants widely used in machine learning, many other functions and algorithms inherently solve our problem of signal enhancement: Quantizers used in digital communication systems and regenerative repeaters have the goal to preserve the information-carrying (discrete) RV as much as possible while removing continuous-valued noise. Bandpass filters (although not yet directly tractable using our notions of relevant information loss) remove out-of-band noise while leaving the information signal unaltered. And finally, methods for dimensionality reduction (such as the principle component analysis) remove redundancy and irrelevance while trying to preserve the interesting part of the observation. In the next section we focus on exactly this class of methods.

\section{Optimality of the PCA in terms of Relevant Information Loss}
\renewcommand{\eye}{\underline{\mathbf{I}}}
\label{sec:PCA}
We now analyze the relevant information loss of the principal component analysis (PCA). In this regard, we initially follow the reasoning of Plumbley~\cite{Plumbley_TN} who showed that the PCA in some cases minimizes the relevant information loss. We then generalize Plumbley's results and even show that this loss can vanish. To this end, we define the PCA as
\begin{equation}
 \Yvec_M=\gvec(\Xvec)=\eye_M\Yvec=\eye_M\Wmat^T\Xvec\label{eq:PCA}
\end{equation}
where $\Xvec$ is an $N$-dimensional continuous-valued input RV, $\Yvec_M$ an $M$-dimensional output RV (composed of the first $M$ elements of $\Yvec$), and $\Wmat$ the matrix of eigenvectors of the covariance matrix of $\Xvec$, $\Cx=\expec{\Xvec\Xvec^T}$. Thus, $\det{\Wmat}=1$ and $\Wmat^{-1}=\Wmat^T$. Furthermore, let $\eye_M$ be an $(M\times N)$-matrix with ones in the main diagonal. The PCA is here used for dimensionality reduction and, assuming that one has perfect knowledge of the rotation matrix $\Wmat$, the \emph{relative} information loss equals $\frac{N-M}{N}$, while the absolute information loss is infinite~\cite{Geiger_RILPCA_arXiv}.

It is known that the PCA minimizes the mean squared error for a reconstruction $\Xvec_M=\Wmat\eye_M^T\Yvec_M$ of $\Xvec$. In addition to that, as Linsker pointed out in~\cite{Linsker_Infomax}, given that $\Xvec$ is an observation of a Gaussian RV $\Svec$ corrupted by Gaussian noise, the PCA maximizes the mutual information $\mutinf{\Svec;\Yvec}$. For non-Gaussian $\Svec$ (and Gaussian noise with spherical symmetry), the PCA not only provides an upper bound on the mutual information, but also an upper bound on the relevant information loss~\cite{Plumbley_TN,Deco_ITNN}.

We now want to generalize the results such that non-iid and non-Gaussian noise is taken into account as well. In particular, we show that the relevant information loss can indeed vanish under some circumstances. Let
\begin{equation}
 \Xvec = \Svec+\Nvec
\end{equation}
where $\Svec$ and $\Nvec$ are the relevant information and the noise, respectively, with covariance matrices $\Cz$ and $\Cn$. We further assume that $\Svec$ and $\Nvec$ are independent and, as a consequence, $\Cx=\Cz+\Cn$. Let $\{\lambda_i\}$, $\{\mu_i\}$, and $\{\nu_i\}$ be the sets of eigenvalues of $\Cx$, $\Cn$, and $\Cz$, respectively. Let the eigenvalues be ordered descendingly, i.e.,
\begin{equation}
 \lambda_1\geq \lambda_2\geq\dots\geq\lambda_N.
\end{equation}
We now write
\begin{equation}
 \Yvec_M = \eye_M\Wmat^T\Xvec = \eye_M\Wmat^T\Svec+\eye_M\Wmat^T\Nvec = \tilde{\Svec}_M+\tilde{\Nvec}_M.
\end{equation}
As mentioned before, $\Yvec_M$ is composed of the first $M$ elements of the vector $\Yvec$. Conversely, we let $\Yvec_c$, $\tilde{\Svec}_c$, and $\tilde{\Nvec}_c$ denote the last $N-M$ elements of the corresponding vectors. If, as in our case, the orthogonal matrix $\Wmat$ performs the PCA, the covariance matrix of $\Yvec_M$ is a diagonal matrix with the $M$ largest eigenvalues of $\Cx$. We now present

\begin{lem}\label{lem:boundIL}
 For above signal model, the relevant information loss in the PCA is given by
\begin{equation}
 \loss[\Svec]{\Xvec\to\Yvec_M} = \diffent{\Yvec_c|\Yvec_M}-\diffent{\tilde{\Nvec}_c|\tilde{\Nvec}_M}.
\end{equation}
\end{lem}

\begin{proof}
\ifthenelse{\arxiv=1}{See Appendix.}{See~\cite{Geiger_Relevant_arXiv}.}

\end{proof}

Before proceeding, we need the following
\begin{definition}\label{def:neg}
 Let $X$ and $Y$ be continuous RVs with arbitrary continuous (joint) distribution, and let $X_G$ and $Y_G$ be Gaussian RVs with same (joint) first and second moments. We define the conditional divergence
\begin{equation}
 J(X|Y) = \kld{X|Y}{X_G|Y_G} = \diffent{X_G|Y_G}-\diffent{X|Y}.
\end{equation}
\end{definition}

Clearly, this quantity inherits all its properties from the Kullback-Leibler divergence, e.g., non-negativity, and can be considered as a measure of Gaussianity. \ifthenelse{\arxiv=1}{We want to point out that despite the similarity, this quantity is not to be confused with \emph{negentropy}. For negentropy, instead of using a jointly Gaussian distribution for $(X_G,Y_G)$ for computing $\diffent{X_G|Y_G}$ one uses a Gaussian distribution for $X_G|Y=y$ to compute $\diffent{X_G|Y=y}$ before taking the expectation w.r.t. the true marginal distribution of $Y$.

}{} While most results about the optimality of the PCA are restricted to the Gaussian case, we now employ conditional divergence to generalize some of these results. In particular, we maintain

\begin{thm}\label{thm:minbound}
 Let $\Xvec=\Svec+\Nvec$, with $\Nvec$ and $\Svec$ independent, and let $\Yvec_M$ be obtained by performing dimensionality-reducing PCA. If $\Nvec$ is iid (i.e., $\Cn$ is a scaled identity matrix) and more Gaussian than $\Yvec$ in the sense
\begin{equation}
 J(\tilde{\Nvec}_c|\tilde{\Nvec}_M)\leq J(\Yvec_c|\Yvec_M)
\end{equation}
then the PCA minimizes the Gaussian upper bound on the relevant information loss $\loss[\Svec]{\Xvec\to\Yvec_M}$.
\end{thm}

\begin{proof}
\ifthenelse{\arxiv=1}{See Appendix.}{See~\cite{Geiger_Relevant_arXiv}.}
\end{proof}

This theorem essentially generalizes the result by Plumbley~\cite{Plumbley_TN,Deco_ITNN}, who claimed that the PCA minimizes the Gaussian upper bound on the information loss for spherical Gaussian noise (i.e., for Gaussian $\Nvec$ with $\Cn$ being a scaled identity matrix). In addition to that, it helps justifying our use of information loss instead of information transfer; an upper bound on the latter would not be useful in our signal enhancement problem.

As a next step, we restrict ourselves to relevant information which is concentrated on an $L\leq M$-dimensional subspace, but drop the requirement that $\Nvec$ is iid. We still assume, however, that $\Cn$ (and, thus, $\Cx$) is full rank. Note that due to these assumptions $\lambda_i>0$ and $\mu_i>0$ for all $i$, while $\nu_i=0$ for $i>L$. \ifthenelse{\arxiv=1}{We are now ready to state}{Employing Cauchy's interlacing inequalities and Weyl's inequalities for eigenvalues we can now prove}

\begin{thm}[Bounds for the PCA]\label{thm:boundsLossPCA}
 Assume that $\Svec$ has covariance matrix $\Cz$ with at most rank $M$, and assume that $\Nvec$ is independent of $\Svec$ and has (full-rank) covariance matrix $\Cn$. Let further $\Nvec$ be more Gaussian than $\Yvec$ in the sense
\begin{equation}
 J(\tilde{\Nvec}_c|\tilde{\Nvec}_M)\leq J(\Yvec_c|\Yvec_M)
\end{equation}
where $\Yvec_M$ is obtained by employing the PCA for dimensionality reduction. Then, the relevant information loss is bounded from above by
\begin{equation}
\loss[\Svec]{\Xvec\to\Yvec_M}\leq \frac{1}{2}\ln\left(\prod_{i=M+1}^{N} \frac{\mu_{1}}{\mu_{i}}\right)
\end{equation}
where $\{\mu_i\}$ is the set of decreasing eigenvalues of $\Cn$ and where $\ln$ is the natural logarithm.
\end{thm}

\begin{proof}
\ifthenelse{\arxiv=1}{See Appendix.}{See~\cite{Geiger_Relevant_arXiv}.}
\end{proof}

\ifthenelse{\arxiv=1}{Note that this upper bound is non-negative, since -- by assumption on the ordering of the eigenvalues -- any term in the product cannot be smaller than one.


We finally show that there are cases where, despite the fact that we reduce the dimensionality of the data, we are still able to preserve all of the relevant information:
\begin{cor}\label{cor:PCALoss}
Assume that $\Svec$ has covariance matrix $\Cz$ with at most rank $M$, and assume that $\Nvec$ is zero-mean Gaussian noise independent of $\Svec$ with covariance matrix $\sigma_\Nvec^2\eye$. Let $\Yvec_M$ be obtained by employing the PCA for dimensionality reduction. Then, the relevant information loss $\loss[\Svec]{\Xvec\to\Yvec_M}$ vanishes.
\end{cor}

\begin{proof}
 The proof follows from the fact that as $\Cn=\sigma_\Nvec^2\eye$, all eigenvalues $\mu_1=\dots=\mu_N=\sigma_\Nvec^2$.
\end{proof}

As mentioned before, due to dimensionality reduction the absolute information loss $\loss{\Xvec\to\Yvec_M}$ is infinite; a direct consequence is that the irrelevant information loss, $\loss[\Xvec|\Svec]{\Xvec\to\Yvec_M}$, is infinite as well. Given the assumptions of Corollary~\ref{cor:PCALoss} holds, the PCA is a good solution to our signal enhancement problem.}
{
A direct consequence of this theorem is that in case of independent, Gaussian noise with $\Cn=\sigma_\Nvec^2\eye$ and all eigenvalues $\mu_i$ being equal, dimensionality reduction is possible without incurring relevant information loss.
}

\section{Examples}
\label{sec:examples}
We now try to illustrate a few of our results with the help of some examples. In particular, we analyze a communication channel with binary input and a dimensionality-reducing PCA for non-Gaussian data. Unless otherwise noted, $\log$ denotes the binary logarithm.

\subsection{Communication channel with uniform noise}
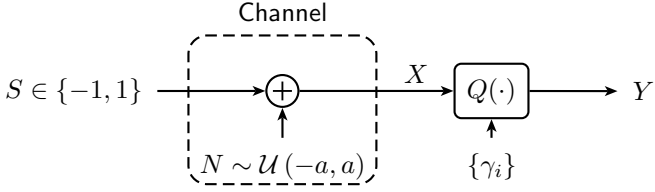
\begin{figure}[t]
 \centering  
\begin{pspicture}[showgrid=false](0,0.5)(8,3.25)
  \psset{style=RoundCorners}
    \pssignal(0.5,2){s}{$S\in\{-1,1\}$}
    \pscircleop(3.25,2){oplus}
    \pssignal(3.25,1){n}{$N\sim\unif{-a}{a}$}
    \psfblock[framesize=1 0.75](6,2){q}{$Q(\cdot)$}
    \pssignal(6,1){quant}{$\{\gamma_i\}$}
    \pssignal(8,2){y}{$Y$}
    \fnode[style=Dash, framesize=2.5 2](3.25,1.75){box}
    \nput{90}{box}{$\mathsf{Channel}$}
    \ncline[style=Arrow]{s}{oplus}
    \ncline[style=Arrow]{n}{oplus}\ncline[style=Arrow]{quant}{q}
  \ncline[style=Arrow]{q}{y} 	
  \ncline[style=Arrow]{oplus}{q}
  \put(4.85,2.1){$X$}
\end{pspicture}
\caption{Digital Communication System: The input signal $S$ is uniformly distributed on $\{-1,1\}$, the channel adds uniform noise to the signal. The channel output, $X$, is quantized using a set of thresholds $\{\gamma_i\}$ to obtain the decision variable $Y$.}
\label{fig:commsys}
\end{figure}
Consider the digital communication system depicted in Fig.~\ref{fig:commsys}. Let the data symbols $S$ be uniformly distributed on $\{-1,1\}$, thus $\ent{S}=1$. Let further $N$ be a noise signal which, for the sake of simplicity, is uniformly distributed on $[-a,a]$, $a>1$. Clearly, $X=S+N$ is a continuous-valued signal with infinite entropy. During quantization an infinite amount of information is lost, $\loss{X\to Y}=\infty$~\cite{Geiger_ILStatic_IZS}. The main point here is that most of the information at the input of the quantizer is information about the noise signal $N$. We will now make this statement precise.

With the differential entropy of $X$ given by
\begin{equation}
 \diffent{X} = \frac{1}{a}\log 4a+\frac{a-1}{a}\log 2a
\end{equation}
and with $\diffent{X|S}=\log 2a$ we get the amount of relevant information available at the input of the quantizer:
\begin{equation}
 \mutinf{X;S} = \frac{1}{a}
\end{equation}
Choosing now a single quantizer threshold $\gamma_1=0$, i.e., $Y=\mathrm{sgn}({X})$ we obtain a binary symmetric channel with cross-over probability $\perr=\frac{a-1}{2a}$. With $\binent{\cdot}$ being the binary entropy function the mutual information thus computes to $\mutinf{Y;S}=1-\binent{\perr}$ (see, e.g.,~\cite{Cover_Information2}) and we obtain a relevant information loss of
\begin{equation}
 \loss[S]{X\to Y} = \binent{\frac{a-1}{2a}} - \frac{a-1}{a}.
\end{equation}

Conversely, if we use two quantizer thresholds $\gamma_1=1-a$ and $\gamma_2=a-1$ we obtain a ternary output RV $Y$. Interpreting any value $\gamma_1\leq X\leq\gamma_2$ as an erasure, we obtain a binary erasure channel with erasure probability $\frac{a-1}{a}$. The corresponding mutual information is computed as $\mutinf{Y;S}=\frac{1}{a}$ (see, e.g.,~\cite{Cover_Information2}), and the relevant information loss vanishes\ifthenelse{\arxiv=2}{\footnote{Preliminary simulations suggest that a similar behavior can also be observed for the Gaussian channel. For all (symmetric) thresholds for a ternary quantizer (leading to a BEC approximation), it seems that there exists an SNR value from which onward ternary quantization outperforms binary quantization.}}{}.

As this example shows, while our system actually destroys an infinite amount of information (to be precise, exactly 100\% of the available information~\cite{Geiger_RILPCA_arXiv}), the relevant information loss can still be zero. Signal enhancement, in the sense of removing irrelevant information, was thus successful.

\ifthenelse{\arxiv=1}{
In this example it is tempting to identify the noise variable $N$ with the irrelevant information $X|S$. However, this not necessarily leads to a correct result, as we show next.

To this end, we substitute the quantizer in Fig.~\ref{fig:commsys} by a magnitude device, i.e., $Y'=|X|$. By the fact that the probability density function of $X$ has even symmetry, it can be shown that $\loss{X\to Y'}=\ent{X|Y'}=1$ (e.g.,~\cite{Geiger_ISIT2011arXiv}). One can further show that, since the marginal distribution of $Y'$ coincides with the conditional distributions $Y|S=-1$ and $Y|S=1$, the mutual information $\mutinf{Y';S}=0$. Thus, $\loss[S]{X\to Y'}=\frac{1}{a}$. As an immediate consequence, $\loss[X|S]{X\to Y'}=\frac{a-1}{a}$.

Let us now determine the information loss relevant w.r.t. $N$. By showing that $\loss[N]{X\to Y'}\neq\loss[X|S]{X\to Y'}$ we argue that noise and irrelevant information are not necessarily identical. Observe that
\begin{IEEEeqnarray}{RCL}
 \loss[N]{X\to Y'} &=& \mutinf{X;N|Y'}\\
&=& \ent{X|Y'} - \ent{X|N,Y'}\\
&=& 1-\ent{S+N|N,Y'}\\
&=& 1-\ent{S|N,|S+N|}.
\end{IEEEeqnarray}
Given we know $N$, $S$ is uncertain only if $|S+N|$ yields the same value for both $S=1$ and $S=-1$. In other words, we have to require that $|N-1|=|N+1|$. Squaring both sides translates this to requiring $N=-N$, which is fulfilled only for $N=0$. Since $\Prob{N=0}=0$ ($N$ is a continuous RV), we automatically obtain
\begin{equation}
 \loss[N]{X\to Y'} = 1 \neq \frac{a-1}{a} =\loss[X|S]{X\to Y'}.
\end{equation}

Indeed, the reason why we cannot identify the noise with the irrelevant information can be related to $N$ and $S$ not being conditionally independent given $X$. We summarize these findings in Table~\ref{tab:ex1}.

\begin{table}
 \caption{Information loss (relevant and irrelevant) for the communication channel with uniform noise.}
\label{tab:ex1}
\centering
\begin{tabular}{c|c|c}
 Loss & $Y=\mathrm{sgn}({X})$ & $Y=|X|$\\
\hline\hline
$\loss{X\to Y}$ & $\infty$ & 1\\\hline
$\loss[S]{X\to Y}$ & $\binent{\frac{a-1}{2a}} - \frac{a-1}{a}$ & $\frac{1}{a}$ \\
$\loss[X|S]{X\to Y}$ & $\infty$ & $\frac{a-1}{a}$\\\hline
$\loss[N]{X\to Y}$ & $\infty$ & 1\\
$\loss[X|N]{X\to Y}$& $\frac{a-1}{2a}$ & 0
\end{tabular}
\end{table}

}{}

\subsection{PCA with non-Gaussian data}
Assume we observe two independent data sources, $S_1$ and $S_2$, with three sensors which are corrupted by independent, unit-variance Gaussian noise $N_i$, $i=1,2,3$. Our sensor signals shall be defined as
\begin{IEEEeqnarray}{RCL}
 X_1&=& S_1+N_1,\\X_2&=&S_1+S_2+N_2,\text{ and}\\X_3&=&S_2+N_3.
\end{IEEEeqnarray}
We assume further that our data sources have variances $\sigma_1^2$ and $\sigma_2^2$ and are non-Gaussian, but that they still can be described by a (joint) probability density function. We obtain the covariance matrix of $\Xvec=[X_1,X_2,X_3]^T$ as
\begin{equation}
 \Cx = \left[\begin{array}{ccc}
              \sigma_1+1  & \sigma_1 &0\\\sigma_1&\sigma_1+\sigma_2+1&\sigma_2\\ 0&\sigma_2&\sigma_2+1 
             \end{array}
\right].
\end{equation}
Performing the eigenvalue decomposition yields three eigenvalues,
\begin{equation}
\{\lambda_1,\lambda_2,\lambda_3\} = \left\{ \sigma_1+\sigma_2+1+C, \sigma_1+\sigma_2+1-C, 1\right\}
\end{equation}
where $C=\sqrt{\sigma_1^2+\sigma_2^2-\sigma_1\sigma_2}$.

We now reduce the dimension of the output vector $\Yvec$ from $N=3$ to $M=2$ by dropping the component corresponding to the smallest eigenvalue. Using Lemma~\ref{lem:boundIL} we now give an upper bound on the relevant loss, $\loss[\Svec]{\Xvec\to\Yvec_M}$. To this end, note that $\Nvec$ is iid Gaussian, and thus $\diffent{\tilde{\Nvec}_c|\tilde{\Nvec}_M}=\diffent{\tilde{\Nvec}_c}$. By assumption, $\Cn=\eye$, and by the orthogonality of the transform,
\begin{equation}
 \diffent{\tilde{\Nvec}_c} = \frac{1}{2}\ln(2\pi\e{}).
\end{equation}
Conversely, we know that
\begin{equation}
 \diffent{\Yvec_c|\Yvec_M}\leq\diffent{\Yvec_{c,G}|\Yvec_{M,G}}=\diffent{\Yvec_{c,G}}=\frac{1}{2}\ln(2\pi\e{}\lambda_3).
\end{equation}
With $\lambda_3=1$ and Lemma~\ref{lem:boundIL} we thus get
\begin{equation}
 \loss[\Svec]{\Xvec\to\Yvec_M}\leq 0.
\end{equation}

The relevant information loss vanishes. Indeed, the eigenvector corresponding to the smallest eigenvalue is given as $\mathbf{p}_3=\frac{1}{\sqrt{3}}[1,-1,1]^T$; thus, for $Y_3$ one would obtain
\begin{equation}
 Y_3=\frac{X_1+X_3-X_2}{\sqrt{3}}=\frac{N_1+N_3-N_2}{\sqrt{3}}.
\end{equation}
Since this component does not contain any signal portion, dropping it does not lead to a loss of relevant information.

Note that in case the noise sources $N_i$ do not have the same variances, the application of PCA may lead to a loss of information, even though the relevant information is still concentrated on a subspace of lower dimensionality. \ifthenelse{\arxiv=1}{We make this precise the next example}{We provide a concrete example in~\cite{Geiger_Relevant_arXiv}}.

\ifthenelse{\arxiv=1}{
\subsection{PCA with large noise variances}
We now again assume that we use three sensors to observe two independent, non-Gaussian data sources which are corrupted by independent Gaussian noise:
\begin{IEEEeqnarray}{RCL}
 X_1&=& S_1+N_1,\\X_2&=&S_2+N_2,\text{ and}\\X_3&=&N_3
\end{IEEEeqnarray}
This time, however, we assume that the data sources have unit variance, and that the variance of noise source $N_i$ is $i$, thus $\{\mu_1,\mu_2,\mu_3\} = \{3,2,1\}$. With this, the covariance matrix of $\Xvec$ is given as
\begin{equation}
 \Cx= \left[\begin{array}{ccc}
              1+1  & 0 &0\\ 0&1+2&0\\ 0&0&3
             \end{array}
\right] = \left[\begin{array}{ccc}
              2  & 0 &0\\0&3&0\\ 0&0&3 
             \end{array}
\right]
\end{equation}
and has eigenvalues $\{\lambda_1,\lambda_2,\lambda_3\} = \{3,3,2\}$.

Since $\Cx$ is already diagonal, the PCA only leads to an ordering w.r.t. the eigenvalues; dropping the component of $\Yvec$ corresponding to the smallest eigenvalue, i.e., dimensionality reduction from $N=3$ to $M=2$, yields $Y_1=N_3$ and $Y_2=S_2+N_2$. Since we do not have access to $S_1$ anymore, information is lost -- in fact, PCA suggested to drop the most informative component of $\Xvec$ (the one with the highest SNR).

With Lemma~\ref{lem:boundIL} we can now compute the Gaussian upper bound on the relevant information loss: Following the proof of Theorem~\ref{thm:minbound},
\begin{equation}
 \loss[\Svec]{\Xvec\to\Yvec_M} \leq \diffent{Y_{3,G}}-\diffent{N_1} = \frac{1}{2}\ln 2.
\end{equation}
The bound obtained from Theorem~\ref{thm:boundsLossPCA} is loose here, evaluating to $ \loss[\Svec]{\Xvec\to\Yvec_M} <\frac{1}{2}\ln 3$.

Since in this case the noise is not iid, the condition of Theorem~\ref{thm:minbound} is not fulfilled, thus the Gaussian upper bound is not minimized by the PCA. Indeed, by preserving $X_1$ and $X_2$ (and dropping $X_3$) the relevant information loss vanishes.
}{}

\section{Conclusion}
In this work we presented the notion of relevant information loss and analyzed many of its elementary properties. We argued that relevant information loss is a central quantity in the problem of signal enhancement and thus, of system theory in general. A comparison with the literature about machine learning and neural networks revealed that many of the algorithms introduced there can be reformulated as solutions for the signal enhancement problem using the notion of relevant information loss. As an example, we discussed principle component analysis used for dimensionality reduction and derived conditions under which the relevant information loss vanishes.

Future work will concentrate on different blocks of the signal processing chain, such as quantizers, sampling devices, and filters as well as on investigating a possible connection between relevant information loss and state space aggregation.

\ifthenelse{\arxiv=1}{
\appendix
\subsection{Proof of Proposition~\ref{prop:dpi}}
 The proof follows along the same lines as the proof of the data processing inequality for mutual information (see~\cite{Cover_Information2}). We first expand $\mutinf{X;W,V|Y}$ as follows:
\begin{IEEEeqnarray}{RCL}
 \mutinf{X;W,V|Y} &=&\mutinf{X;V|Y}+\mutinf{X;W|V,Y}\\
 &=& \mutinf{X;W|Y}+\mutinf{X;V|W,Y}\\
&=& \mutinf{X;W|Y}
\end{IEEEeqnarray}
where the last line follows from the fact that $V$ and $X$ are conditionally independent given $W$.

Using now Definition~\ref{def:relevantLoss} we obtain
\begin{IEEEeqnarray}{RCL}
 \loss[W]{X\to Y}&=& \mutinf{X;W|Y}\\
&=&\mutinf{X;V|Y}+\mutinf{X;W|V,Y}\\
&=&\loss[V]{X\to Y}+\mutinf{X;W|V,Y}\\
&\geq&\loss[V]{X\to Y}
\end{IEEEeqnarray}
which completes the proof.\endproof

\subsection{Proof of Corollary~\ref{cor:dpi}}
If $f(S)-S-X-Y$ is a Markov chain, the proof follows immediately from Proposition~\ref{prop:dpi}. In case $S-f(S)-X-Y$, one gets with the proof of Proposition~\ref{prop:dpi}
\begin{IEEEeqnarray}{RCL}
 \loss[f(S)]{X\to Y}&=&\loss[S]{X\to Y}+\mutinf{X;f(S)|S,Y}\\&=&\loss[S]{X\to Y}.
\end{IEEEeqnarray}
Thus, in this case Proposition~\ref{prop:dpi} is shown to hold with equality.

Furthermore, using Definition~\ref{def:relevantLoss}, the Corollary follows immediately from~\cite[Thm.~3.7.1]{Pinsker_InfoEngl}.\endproof

\subsection{Proof of Proposition~\ref{prop:cascade}}
 For the proof, note that with Definition~\ref{def:relevantLoss} we get
\begin{IEEEeqnarray}{RCL}
 \loss[S]{X\to Z} &\stackrel{(a)}{=}& \mutinf{S;X,Y|Z}\\
&\stackrel{(b)}{=}& \mutinf{S;X|Y,Z}+\mutinf{S;Y|Z}\\&\stackrel{(a)}{=}& \mutinf{S;X|Y}+\mutinf{S;Y|Z}
\end{IEEEeqnarray}
where $(a)$ is due to the fact that $Y=g(X)$ and $Z=h(Y)$, respectively, and $(b)$ is the chain rule of information.\endproof

\subsection{Proof of Lemma~\ref{lem:boundIL}}
We start by noting that $\loss[\Svec]{\Xvec\to\Yvec_M} = \loss[\Svec]{\Yvec\to\Yvec_M}$ since $\Yvec$ and $\Xvec$ are related by an invertible transform. Thus,
 \begin{IEEEeqnarray}{RCL}
  \loss[\Svec]{\Yvec\to\Yvec_M} &=& \mutinf{\Yvec;\Svec}-\mutinf{\Yvec_M;\Svec}\\
&=& \diffent{\Yvec}-\diffent{\Yvec|\Svec}-\diffent{\Yvec_M}+\diffent{\Yvec_M|\Svec}\notag\\
&=& \diffent{\Yvec_c|\Yvec_M}-\diffent{\Wmat^T\Svec+\Wmat^T\Nvec|\Svec}\notag\\ &&{}+\diffent{\eye_M\Wmat^T\Svec+\eye_M\Wmat^T\Nvec|\Svec}\\
&\stackrel{(a)}{=}& \diffent{\Yvec_c|\Yvec_M}-\diffent{\Wmat^T\Nvec}+\diffent{\eye_M\Wmat^T\Nvec}\notag\\
&=&\diffent{\Yvec_c|\Yvec_M}-\diffent{\tilde{\Nvec}_c|\tilde{\Nvec}_M}
 \end{IEEEeqnarray}
where $(a)$ is due to independence of $\Svec$ and $\Nvec$. This completes the proof.\endproof

\subsection{Proof of Theorem~\ref{thm:minbound}}
By Lemma~\ref{lem:boundIL} and Definition~\ref{def:neg},
\begin{IEEEeqnarray}{RCL}
 \loss[\Svec]{\Xvec\to\Yvec_M} &=&\diffent{\Yvec_c|\Yvec_M}-\diffent{\tilde{\Nvec}_c|\tilde{\Nvec}_M}\\
&=& \diffent{\Yvec_{c,G}|\Yvec_{M,G}} - J(\Yvec_c|\Yvec_M)\notag\\&&{}-\diffent{\tilde{\Nvec}_{c,G}|\tilde{\Nvec}_{M,G}}+J(\tilde{\Nvec}_c|\tilde{\Nvec}_M)\\
&\leq& \diffent{\Yvec_{c,G}|\Yvec_{M,G}}-\diffent{\tilde{\Nvec}_{c,G}|\tilde{\Nvec}_{M,G}}\\
&\stackrel{(a)}{=}& \diffent{\Yvec_{c,G}}-\diffent{\tilde{\Nvec}_{c,G}|\tilde{\Nvec}_{M,G}}\label{eq:line1}\\
&\stackrel{(b)}{=}& \diffent{\Yvec_{c,G}}-\diffent{\tilde{\Nvec}_{c,G}}\\
&=& \frac{1}{2}\ln \left(\prod_{i=M+1}^{N}\frac{\lambda_i}{\mu}\right).
\end{IEEEeqnarray}
Here, $(a)$ is due to the fact that the PCA decorrelates the output data $\Yvec$ and thus leads to independence of $\Yvec_{c,G}$ and $\Yvec_{M,G}$ (in the sense of Definition~\ref{def:neg}). By similar reasons $(b)$ follows from the fact that $\Nvec$ is iid ($\Cn$ is a scaled identity matrix, with all eigenvalues being equal $\mu_i=\mu$). Since the PCA preserves the $M$ coordinates corresponding to the $M$ largest eigenvalues of $\Cy$, the last line (obtained with~\cite[Thm. 8.4.1]{Cover_Information2} and~\cite[Fact 5.10.14]{Bernstein_Matrix}) represents the smallest Gaussian upper bound and completes the proof.\endproof

\subsection{Proof of Theorem~\ref{thm:boundsLossPCA}}
We recapitulate~\eqref{eq:line1} from the proof of Theorem~\ref{thm:minbound} and get
\begin{IEEEeqnarray}{RCL}
  \loss[\Svec]{\Xvec\to\Yvec_M} &\leq& \diffent{\Yvec_{c,G}}-\diffent{\tilde{\Nvec}_{c,G}|\tilde{\Nvec}_{M,G}}\\
&=& \diffent{\Yvec_{c,G}} - \diffent{\tilde{\Nvec}_G}+\diffent{\tilde{\Nvec}_{M,G}}.
\end{IEEEeqnarray}
With~\cite[Thm. 8.4.1]{Cover_Information2} and~\cite[Fact 5.10.14]{Bernstein_Matrix} we get
\begin{equation}
 \diffent{\tilde{\Nvec}_G}=\frac{1}{2}\ln\left((2\pi\e{})^N\prod_{i=1}^N \mu_i\right)
\end{equation}
and
\begin{equation}
 \diffent{\Yvec_{c,G}} = \frac{1}{2}\ln\left((2\pi\e{})^{N-M}\prod_{i=M+1}^N \lambda_i\right).
\end{equation}
If now $\Cntm$ denotes the $(M\times M)$-covariance matrix of $\tilde{\Nvec}_{M}$ (and, thus, of $\tilde{\Nvec}_{M,G}$) and $\{\tilde{\mu}_i\}$ the set of eigenvalues of $\Cnt$, we obtain
\begin{equation}
 \loss[\Svec]{\Xvec\to\Yvec_M}\leq\frac{1}{2}\ln\left(\frac{\prod_{i=M+1}^N \lambda_i \prod_{i=1}^M \tilde{\mu}_i}{\prod_{i=1}^N \mu_i}\right).
\end{equation}

We now complete the proof by providing upper bounds on the eigenvalues in the numerator. It is easy to verify that $\Cnt$ is the top left principal submatrix of $\Wmat^T\Cn\Wmat$ (which, by the orthogonality of $\Wmat$ has the same eigenvalues as $\Cn$). As a consequence, we can employ Cauchy's interlacing inequality~\cite[Thm. 8.4.5]{Bernstein_Matrix}:
\begin{equation}
 \mu_{i+N-M}\leq\tilde{\mu}_i\leq\mu_i
\end{equation}
The second bound, $\lambda_{i}\leq\mu_1$, is derived from Weyl's inequality~\cite[Thm. 8.4.11]{Bernstein_Matrix}
\begin{IEEEeqnarray}{RCL}
 \lambda_{i} \leq \nu_i +\mu_1
\end{IEEEeqnarray}
and by noticing that $\nu_j=0$ for all $j>M$.

Combining all this, we obtain as an upper bound on the information loss
\begin{IEEEeqnarray}{RCL}
 \loss[\Svec]{\Xvec\to \Yvec_M} &\leq& \frac{1}{2}\ln\left(\frac{\prod_{i=M+1}^N \lambda_i \prod_{i=1}^M \tilde{\mu}_i}{\prod_{i=1}^N \mu_i}\right)\\
&\leq& \frac{1}{2}\ln\left(\frac{\prod_{i=M+1}^N \mu_1 \prod_{i=1}^M \mu_i}{\prod_{i=1}^N \mu_i}\right)\\
&=& \frac{1}{2}\ln\left(\prod_{i=M+1}^{N} \frac{\mu_1}{\mu_{i}}\right).
\end{IEEEeqnarray}
This completes the proof.\endproof
}{}

\bibliographystyle{IEEEtran}
\bibliography{IEEEabrv,/afs/spsc.tugraz.at/project/IT4SP/1_d/Papers/InformationProcessing.bib,%
/afs/spsc.tugraz.at/project/IT4SP/1_d/Papers/ProbabilityPapers.bib,%
/afs/spsc.tugraz.at/user/bgeiger/includes/textbooks.bib,%
/afs/spsc.tugraz.at/user/bgeiger/includes/myOwn.bib,%
/afs/spsc.tugraz.at/user/bgeiger/includes/UWB.bib,%
/afs/spsc.tugraz.at/project/IT4SP/1_d/Papers/InformationWaves.bib,%
/afs/spsc.tugraz.at/project/IT4SP/1_d/Papers/ITBasics.bib,%
/afs/spsc.tugraz.at/project/IT4SP/1_d/Papers/HMMRate.bib,%
/afs/spsc.tugraz.at/project/IT4SP/1_d/Papers/ITAlgos.bib}
\end{document}